\documentclass{article}

\usepackage[T1]{fontenc}
\usepackage{lmodern}
\usepackage{rsfso}

\usepackage{graphicx}
\usepackage{adjustbox}
\usepackage{pdflscape}
\usepackage{multirow}
\usepackage{amsmath,amssymb,amsfonts}
\usepackage{amsthm}
\usepackage[title]{appendix}
\usepackage{xcolor}
\usepackage{textcomp}
\usepackage{manyfoot}
\usepackage{booktabs}
\usepackage{longtable}
\usepackage{algorithm}
\usepackage{algorithmicx}
\usepackage{algpseudocode}
\usepackage{authblk}
\usepackage{listings}
\usepackage{enumitem}
\usepackage{multicol}
\usepackage{makecell}
\usepackage{mathtools}
\usepackage{placeins}
\usepackage{microtype}
\usepackage{xurl}
\usepackage[
  unicode,
  pdfencoding=auto,
  psdextra,
  hidelinks
]{hyperref}
\usepackage{cleveref}
\usepackage{subcaption}

\lstdefinestyle{prompt}{
  basicstyle=\ttfamily\footnotesize,
  breaklines=true,
  breakatwhitespace=false,
  columns=fullflexible,
  keepspaces=true,
  showstringspaces=false,
  upquote=true
}

\theoremstyle{plain}
\newtheorem{theorem}{Theorem}

\theoremstyle{definition}

\date{}
\theoremstyle{remark}

\title{Improved lower bounds for the Shannon capacity of odd cycles}
\author[1]{Nathaniel Itty\thanks{Email: \texttt{nathanielitty@gmail.com}}}
\author[2]{Christopher D. Rosin\thanks{Email: \texttt{christopher.rosin@gmail.com}}}
\author[1]{Chase Carstensen}
\author[1]{Daniel Reichman}

\affil[1]{Worcester Polytechnic Institute}
\affil[2]{Constructive Codes}

\begin{document}
\maketitle

\begin{abstract}
The Shannon capacity $\Theta(G)$ of a graph $G$ quantifies the maximum rate at which information can be transmitted with zero error over a noisy channel. It is lower bounded by $\alpha(G^d)^{1/d}$ for any $d$, where $\alpha(G^d)$ is the independence number of the $d$-th strong product of $G$.
We construct independent sets of size $134753$ in $C_7^{10}$, $21909$ in $C_{11}^{6}$, $62530$ in $C_{13}^{6}$, and $8076974$ in $C_{15}^{8}$, improving the best known lower bounds for the Shannon capacity of these graphs to $\Theta(C_7)\geq 134753^{1/10}>3.258020$, $\Theta(C_{11})\geq 21909^{1/6}>5.289773$, $\Theta(C_{13})\geq 62530^{1/6}>6.300109$, and $\Theta(C_{15})\geq 8076974^{1/8}>7.301399$. We also improve the best known lower bounds on the independence numbers of several individual strong products of odd cycles that do not improve the Shannon capacity lower bound. The constructions were discovered through iterative interactions with a Large Language Model (LLM), illustrating the potential of LLMs for finding explicit combinatorial constructions.
\end{abstract}

\section{Introduction}
Let $G=(V,E)$ be an undirected graph. The strong $d$-product of $G$, denoted by $G^d$ has vertex-set $V^d$. Two distinct vertices $(u_1,\ldots u_d)$ and $(v_1,\ldots v_d)$ in $V^d$ are connected by an edge in $G^d$ if and only if for every $i, u_i=v_i$ or $(u_i,v_i) \in E$. Recall that for a graph $H,$ $\alpha(H)$ is the size of the largest independent set in $H$. The \emph{Shannon Capacity} of $G$, denoted $\Theta(G)$  is defined as
\[\Theta(G)=\sup_{d \in \mathbb{N}}\alpha(G^d)^{1/d}.\]

It is well known that for any two positive integers $d_1,d_2$, $\alpha(G^{d_1+d_2}) \geq \alpha(G^{d_1})\alpha(G^{d_2})$, which implies that
$\Theta(G)=\lim_{d \rightarrow \infty}\sqrt[d]{\alpha(G^d)}.$

The Shannon capacity was introduced by Shannon~\cite{shannon1956zero} in 1956 to quantify the communication capacity of a noisy channel. It has been extensively studied: for an excellent recent survey we refer the reader to~\cite{lavi2025advances}. For the even cycle with $p$-vertices, it is straightforward to prove that $\Theta(C_p)=p/2$ (e.g.,~\cite{polak2019new}). The celebrated work of Lov\'asz~\cite{lovasz1979shannon} introduced the $\vartheta$-function and used it to prove that for the $5$-cycle $C_5$, $\Theta(C_5)=\sqrt{5}.$  For odd cycles $C_{2r+1}$ with $r>2$ the value of $\Theta(C_{2r+1})$ is not known. Even for the simplest case, determining the value $\Theta(C_7)$ is considered a notoriously difficult problem that is still open despite extensive effort~\cite{polak2019new,vesel2002improved,mathew2017new,guruswami2021linear}. The best known lower bound of $\Theta(C_7)$ was found~\cite{polak2019new} by clever construction of an independent set of size $367$ in $C_7^5$. The construction in~\cite{polak2019new} combines combinatorial ideas along with combinatorial search using an optimization solver (Gurobi). This immediately implies that 
$\Theta(C_7)\geq (367)^{1/5}> 3.2578.$  

Here we provide modest improvements on the best known lower bound of $\Theta(C_7)$, $\Theta(C_{11})$, and $\Theta(C_{15})$, and a somewhat larger improvement for $\Theta(C_{13})$:

\begin{theorem}~\label{thm:main}
We have that $\Theta(C_7) > 3.258020$, $\Theta(C_{11}) > 5.289773$, $\Theta(C_{13}) > 6.300109$, and $\Theta(C_{15})>7.301399$.
\end{theorem}
\begin{proof}
This follows from the existence of an independent set of size $134753$ in $C_7^{10}$, an independent set of size $21909$ in $C_{11}^{6}$, an independent set of size $62530$ in $C_{13}^6$, and an independent set of size $8076974$ in $C_{15}^{8}$.  The independent sets can be found in the GitHub repository. A summary of the resulting lower bounds and comparison to previous lower bounds can be found in Table~\ref{tab:shannon-capacity-bounds}.
\end{proof}

\begin{table}[htbp]
\centering
\caption{Our improvements to the lower bounds on the Shannon Capacity for $C_7$, $C_{11}$, $C_{13}$, and $C_{15}$.}
\label{tab:shannon-capacity-bounds}
\renewcommand{\arraystretch}{1.15}
\setlength{\tabcolsep}{9pt}
\begin{tabular}{ccc}
\toprule
Graph & Previous lower bound & New lower bound \\
\midrule
$C_7$
& $367^{1/5} > 3.257865$~\cite{polak2019new}
& $134753^{1/10} > 3.258020$ \\

$C_{11}$
& $148^{1/3} > 5.289572$~\cite{baumert1971combinatorial}
& $21909^{1/6} > 5.289773$ \\

$C_{13}$
& $247^{1/3} > 6.274305$~\cite{baumert1971combinatorial,bohman2013independence}
& $62530^{1/6} > 6.300109$ \\

$C_{15}$
& $2842^{1/4} > 7.301397$~\cite{de2024asymptotic}
& $8076974^{1/8} > 7.301399$ \\
\bottomrule
\end{tabular}
\end{table}

We also improve several known lower bounds on the independence numbers of products of odd cycles that do not lead to improved lower bounds on the Shannon capacity. For example, we prove that $\alpha(C_{15}^3) \geq 383$ improving upon the previous bound in~\cite{codenotti2003some}. We report these in Appendix~\ref{appendix:additional-bounds} as they could be instrumental for future improvements.

\paragraph{Statement on AI use}
Theorem~\ref{thm:main} was obtained using multiple prompts to ChatGPT-5.6 Sol Pro. 
The correctness of the constructions (feasibility of the independent sets constructed)
was verified by the authors. The paper was written by the authors (LLMs were used for minor edits to improve grammar and style).

Perhaps surprisingly, ChatGPT-5.6 Sol Pro produced independent sets that were not found by the search heuristics manually implemented by the authors, including simulated annealing. Even local search algorithms that were constructed with generative AI (CPro1~\cite{rosin2025using}) failed to reach the improved lower bounds reported in Theorem~\ref{thm:main} despite more than 3 months of repeated attempts. We believe that this points to the importance of mathematical knowledge in improving lower bounds for the Shannon capacity of odd cycles. Interestingly several state of the art lower bounds for $C_7$ and $C_{15}$ ~\cite{polak2019new,de2024asymptotic} rely on a combination of heuristic search coupled with expert knowledge. In addition, it seems that ChatGPT is able to leverage existing known independent sets in strong products of graphs and modify them using ``small'' changes to improve upon current best lower bounds. Studying the mechanisms that allow it to surpass search heuristics for finding independent sets remains an interesting future research direction. 

\section{Finding independent sets in strong products of odd cycles}
The idea of improving the lower bounds for the Shannon capacity of $G$ by finding successively larger independent sets in suitable strong products
of $G$ is a natural and common approach that we follow as well. However, this approach is not without difficulties as it requires finding a quantity that is NP-hard to compute even approximately~\cite{haastad1999clique} in general, for graphs that reach thousands or more vertices for a modest value of the exponent $d$. The exponential growth rate of the size of the strong product makes it difficult to compute the size of $\alpha(G^d)$. For example, $\alpha(C_7^4)$ is currently not known. It is only known that $\alpha(C_7^4) \geq 108$ as it contains an independent set of this size~\cite{vesel2002improved}. It is a major open problem whether deciding if $\Theta(G)$ is larger than a given threshold is a decidable problem~\cite{alon2006shannon}. Even understanding the independence number of the $3$rd product of odd cycles is challenging. While the exact value of $\alpha(C_7^3)$ is known, only partial results are known about $\alpha(C^3_{2r+1})$ for $r>3$ and devising an exact formula for this quantity is currently open~\cite{bohman2003limit}. 

The recent use of Large Language Models (LLMs) for finding mathematical constructions has been applied to the Shannon capacity of odd cycles. In~\cite{romera2024mathematical} LLMs were used to rediscover $\alpha(C_7^{5}) \geq 367$, state-of-the-art lower bounds for $\alpha(C_9^d)$ for $d=3,..,7$ and improve the best lower bound on $\alpha(C_{11}^{4})$ by finding an independent set of size 754 in $C_{11}^4.$ These results were obtained by a simple greedy algorithm discovered using the FunSearch framework for finding mathematical constructions leveraging algorithms created with a combination of evolutionary algorithms and LLMs. This paradigm was used recently~\cite{zhai2025x} to improve the best known lower bound for $C_{15}^5$ by constructing an independent set of size 19946.  However, none of these prior LLM results led to improved lower bounds on the Shannon capacity of odd cycles. 

\subsection{Our approach}
Searches were conducted through the standard ChatGPT web interface using ChatGPT-5.6 Sol Pro. For each instance, we specified the cycle length, product dimension, best construction known to us, and target cardinality required for an improvement. The model generated search programs, executed them within its environment, and returned resulting independent sets in a specified format.

The prompt can be found in Appendix~\ref{appendix:prompt}.

\section{Constructions}

We represent the vertices of the cycle $C_k$ by $\mathbb{Z}_{k} = \{0,1,\ldots,k-1\}$, and vertices in the strong $n$-product by vectors in $\mathbb{Z}^{n}_{k}$.  Two distinct vertices have an edge if and only if their circular distance is at most 1 in every coordinate of their vectors.

\subsection{\texorpdfstring{7-Cycle $C_{7}$}{7-Cycle C7}}

For the 7-cycle $C_7$, the best current bound comes from the strong $5$-product, for which the largest known independent set has size 367 \cite{polak2019new}, giving $\Theta(C_7) \geq 367^{1/5} \approx 3.257866$.

Let $R$ be the explicit size-367 set of 5-vectors as given in the Appendix of the original publication \cite{polak2019new}.  Note that $R \times R$ would give a set of $367 \times 367 = 134689$ 10-vectors that constitute an independent set in the strong $10$-product.  We are instead going to delete several vectors from $R$ to yield $B$, and then take $B \times B$ augmented with additional vectors derived from $R$, to reach a larger independent set.

Define $r_j$ and $q_j$ according to the table below.  

\begin{table}[htbp]
\centering
\begin{tabular}{c c c}
\hline
$j$ & $r_j$ & $q_j$ \\
\hline
0 & (1,3,4,4,6) & (2,3,5,4,6) \\
1 & (3,4,0,3,5) & (2,4,6,3,5) \\
2 & (5,3,1,3,4) & (5,3,2,3,5) \\
3 & (4,4,6,1,6) & (5,4,6,0,6) \\
4 & (6,0,6,4,5) & (6,1,6,5,5) \\
5 & (0,3,5,6,5) & (6,3,5,0,5) \\
6 & (6,4,3,4,0) & (6,4,2,4,6) \\
7 & (6,4,5,3,2) & (6,5,5,3,1) \\
\hline
\end{tabular}
\end{table}

Let $B$ be a subset of 359 vectors from $R$, defined as $B = R \setminus \{r_j: 0 \leq j < 8\}$.

Taking arithmetic modulo 7, define:
\[X_0 = \{(2-w_1, w_3, w_0, (2-w_2), w_4): w \in R\}\]
Let $X$ be $X_0$ with the vector (2,4,6,3,5) replaced by (1,5,6,3,5).

Define $J_0 = \{0,5,6\}$ and $J_1 = \{1,2,3,4,7\}$.  Let:
\[P_H = \{r_j: j \in J_0\} \cup \{q_j: j \in J_1\}\]
\[P_V = \{q_j: j \in J_0\} \cup \{r_j: j \in J_1\}\]

Let $x \sim y$ indicate that the vertices corresponding to $x$ and $y$ have an edge; the vectors differ by at most 1 (taken modulo 7) in each coordinate.  Define:

\[A = \{x \in X|\exists y \in P_V : x \sim y\}\]
\[D = \{x \in X|\exists y \in P_H : x \sim y\}\]

The set $A$ has 20 vectors and $D$ has 26.
Define functions $h_j(x)$ and $v_j(x)$ for $x \in X$ as follows:

If $j \in J_0$ and $x \in A$, or $j \in J_1$ and $x \in D$, then let $h_j(x) = q_j$.  Otherwise let $h_j(x) = r_j$.

If $j \in J_0$ and $x \in D$, or $j \in J_1$ and $x \in A$, then let $v_j(x) = q_j$.  Otherwise let $v_j(x) = r_j$.

Now define this set of 10-vectors:
\[I = (B \times B) \cup \{(h_j(x),x): x \in X, 0 \leq j < 8\} \cup \{(x,v_j(x)): x \in X, 0 \leq j < 8\}\]

$I$ has $359 \times 359 + 8 \times 367 + 8 \times 367 = 134753$ vectors.  We verified that the vertices given by these vectors constitute an independent set in the strong $10$-product of $C_7$; for any two distinct $x,y \in I$ it is not the case that $x \sim y$.  This raises the lower bound on the Shannon capacity for $C_7$, giving $\Theta(C_{7}) \geq 134753^{1/10} > 3.258020$.

\subsection{\texorpdfstring{11-Cycle $C_{11}$}{11-Cycle C11}}

For the 11-cycle $C_{11}$, the best current bound comes from the strong 3-product, for which the largest known independent set has size 148 \cite{baumert1971combinatorial}, giving $\Theta(C_{11}) \geq 148^{1/3} \approx 5.289572$.  We start from this independent set.

Let $R$ be the size-148 independent set of 3-vectors described in~\cite{baumert1971combinatorial} (see Table~\ref{tab:r148}). Define the following sets of 3-vectors (with arithmetic taken modulo 11):
\begin{align*}
D_L &= \{(0,0,2),(3,0,2),(1,3,2),(3,2,2)\}, \\
D_R &= \{(0,0,2),(3,0,2)\}, \\
B_L &= R \setminus D_L, \\
B_R &= R \setminus D_R, \\
X &= \{(x_0+1,x_1+10,x_2+10):x\in R\}, \\
Y &= \{(x_0+1,x_1+1,x_2+10):x\in R\}, \\
F_A &= \{(1,10,1),(10,0,2)\}, \\
F'_A &= \{(0,1,2),(1,3,2),(2,1,2)\}, \\
F_B &= \{(2,0,4),(4,10,1),(4,10,3), \\
&\qquad (4,1,1),(4,1,3),(0,0,4)\}, \\
F'_B &= \{(1,0,2),(1,2,2),(3,0,2),(3,2,2)\}, \\
F_C &= \{(2,2,1),(2,2,3),(0,2,1),(0,2,3)\}, \\
F'_C &= \{(1,0,3),(1,3,2),(3,0,2),(3,2,2)\}, \\
G_A &= \{(1,1,1)\}, \\
G'_A &= \{(1,1,2)\}, \\
G_B &= \{(1,10,2),(2,2,4),(2,0,4),(3,10,2),(4,1,1), \\
&\qquad (4,1,3),(4,3,1),(4,3,3),(0,2,4),(0,0,4)\}, \\
G'_B &= \{(0,0,2),(2,0,2)\}, \\
G_C &= \{(2,4,1),(2,4,3),(10,2,2),(10,0,2), \\
&\qquad (0,4,1),(0,4,3)\}, \\
G'_C &= \{(1,0,3),(3,0,2)\}.
\end{align*}

Define set-valued functions $H(x)$ and $V(y)$ for $x \in X$ and $y \in Y$ as follows:

If $x \in F_A$ then $H(x)=F'_A$.  If $x \in F_B$ then $H(x)=F'_B$.  If $x \in F_C$ then $H(x)=F'_C$.  For all other values of $x$, $H(x)=D_L$.

If $y \in G_A$ then $V(y)=G'_A$.  If $y \in G_B$ then $V(y)=G'_B$.  If $y \in G_C$ then $V(y)=G'_C$.  For all other values of $y$, $V(y)=D_R$.

Now define this set of $6$-vectors:
\[I = B_L \times B_R \cup \{(h,x):x \in X, h \in H(x)\} \cup \{(y,v): y \in Y, v \in V(y)\}\]
$I$ has:
\begin{align*}
|I|
&= 144\cdot146
 + \bigl(2\cdot3+6\cdot4+4\cdot4+(148-12)\cdot4\bigr) \\
&\quad
 + \bigl(1\cdot1+10\cdot2+6\cdot2+(148-17)\cdot2\bigr) \\
&= 21909.
\end{align*}
vectors.  We verified that the vertices given by these vectors constitute an independent set in the strong $6$-product of $C_{11}$; for any two distinct $x,y \in I$ it is not the case that $x \sim y$.  This raises the lower bound on the Shannon capacity for $C_{11}$, giving $\Theta(C_{11}) \geq 21909^{1/6} > 5.289773$.

\subsection{\texorpdfstring{13-Cycle $C_{13}$}{13-Cycle C13}}

For the 13-cycle $C_{13}$, the best current bound comes from the strong $3$-product, in which the maximum independent set size was found to be exactly 247 \cite{bohman2013independence}, giving $\Theta(C_{13}) \geq 247^{1/3} \approx 6.274305$.  Define:

\[
A=
\begin{pmatrix}
1 & 0 & 0 & 0 & 11 & 12 \\
0 & 1 & 0 & 0 & 9 & 11 \\
0 & 0 & 1 & 7 & 1 & 0 \\
0 & 0 & 0 & 0 & 1 & 0 \\
\end{pmatrix}.
\]

Define a graph $G_A$ with $13^4$ vertices corresponding to $\mathbb{Z}^{4}_{13}$.  Two distinct vertices $s$ and $t$ are adjacent if $\exists d \in \{-1,0,1\}^6 : s-t \equiv Ad \mod 13$.  As part of its computational process during the chat, ChatGPT implemented and executed a randomized local search to find a large independent set in the graph $G_A$.  Let $S$ be the set of 370 4-vectors in this independent set (see Table~\ref{tab:s-table}).  Let
\[I = \{x \in \mathbb{Z}^{6}_{13}: Ax \in S\} \]
with arithmetic modulo 13.

The resulting set $I$ has 62530 6-vectors.  We verified that each pair of vectors of $I$ differs by more than 1 (modulo $13$) in at least one coordinate.  Therefore, the vertices corresponding to $I$ constitute an independent set in the strong $6$-product of cycle $C_{13}$.  This gives $\Theta(C_{13}) \geq 62530^{1/6} > 6.300109$.  

\subsection{\texorpdfstring{15-Cycle $C_{15}$}{15-Cycle C15}}

For the 15-cycle $C_{15}$, the best current bound comes from the strong $4$-product, for which the largest known independent set has size 2842 \cite{de2024asymptotic}, giving $\Theta(C_{15}) \geq 2842^{1/4} \approx 7.301398$.

Let $R_0$ be the explicit size-2842 set of 4-vectors as given in Appendix B of the original publication \cite{de2024asymptotic}.  Delete the vectors shown in the DELETE column of Table~\ref{tab:c15deleteadd}, and add the vectors shown in the ADD column.  Call the resulting set $R$, which also has 2842 vectors.

\begin{table}[h]
\centering
\begin{tabular}{c c}
\toprule
\textbf{DELETE} & \textbf{ADD} \\
\midrule
(13,0,0,9) & (13,0,0,8) \\
(13,14,0,7) & (13,14,0,6) \\
(13,13,0,5) & (13,14,14,4) \\
(13,14,13,4) & (13,14,12,3) \\
(13,14,11,3) & (13,0,10,3) \\
(13,0,9,3) & (13,0,8,2) \\
(13,0,7,2) & (13,0,6,1) \\
(13,0,5,1) & (14,1,4,1) \\
(14,1,3,1) & (14,1,2,0) \\
(14,1,1,0) & (14,1,0,14) \\
(14,1,0,13) & (14,1,0,12) \\
(13,0,0,11) & (13,0,0,10) \\
\bottomrule
\end{tabular}
\caption{Vectors to DELETE from $R_0$ and ADD to make $R$.}
\label{tab:c15deleteadd}
\end{table}

Let $B$ be a subset of 2838 vectors from $R$, defined as $B = R \setminus \{D_i: 0 \leq i < 4\}$ with $D_i$ shown in Table~\ref{tab:c15diqi}.  

For $i \in \{0,2,3\}$ let $T_i = \{(x_0,x_1-1,x_2,x_3): x \in R\}$, and   
for $i=1$ let $T_1 = \{(-x_1+5,x_0-1,-x_2+7,-x_3+3): x \in R\}$, with all arithmetic mod 15.

\begin{table}[h]
\centering
\begin{tabular}{c c c}
\toprule
$i$ & $D_i$ & $Q_i$ \\
\midrule
0 & (1,0,14,9) & (1,1,14,9) \\
1 & (4,0,8,9) & (3,0,8,9) \\
2 & (11,14,1,10) & (11,0,1,10) \\
3 & (14,1,0,14) & (14,2,0,14) \\
\bottomrule
\end{tabular}
\caption{The vectors $D_i$ and $Q_i$.}
\label{tab:c15diqi}
\end{table}

For vectors $x$ and $y$, as usual let $x \sim y$ indicate that the vectors differ by at most 1 (taken mod 15) in each coordinate.

Let $I_0 = B \times B$ be our initial set of 2838*2838=8054244 8-vectors.  We will then augment it with the 8-vectors described below to form $I$, our final set.

For $0 \leq i < 4$ and for each $x \in T_i$, let $j$ be the minimum value such that $x \sim D_j$ or $x \sim Q_j$.  If $j$ doesn't exist (because the condition is never satisfied) then add $(D_i,x)$ to $I$.  Otherwise ($j$ exists) check $W_0$ table entry $(i,j)$ in Table~\ref{tab:c15w0w1}.  The table entry is a triple.  If both $x \sim D_j$ and $x \sim Q_j$ then use the third element of the triple; otherwise if $x \sim D_j$ use the first element of the triple, or if $x \sim Q_j$ use the second element of the triple.  If the chosen element of the triple is D then add $(D_i,x)$ to $I$, if the chosen element of the triple is Q then add $(Q_i,x)$ to $I$, and if the chosen element of the triple is a dash (--) then add nothing to $I$.

Do the same using $W_1$, but add $(x,D_i)$ or $(x,Q_i)$ instead of $(D_i,x)$ or $(Q_i,x)$.  

$W_0$ and $W_1$ each add 11365 vectors to $I$, giving us a total of $8076974$ vectors.  

We verified that the vertices given by these vectors constitute an independent set in the strong 8-product of $C_{15}$; for any two distinct $x,y \in I$ it is not the case that $x \sim y$.  This improves the lower bound on the Shannon capacity for $C_{15}$, giving $\Theta(C_{15}) \geq 8076974^{1/8} > 7.301399$.  

\begin{table}[h]
\centering
\begin{subtable}{0.45\textwidth}
\centering
\begin{tabular}{c|cccc}
 & \multicolumn{4}{c}{$j$} \\
$i$ & 0 & 1 & 2 & 3 \\
\hline
0 & \texttt{DQ–} & \texttt{D–Q} & \texttt{DQ–} & \texttt{DQ–} \\
1 & \texttt{D––} & \texttt{DQ–} & \texttt{–D–} & \texttt{D––} \\
2 & \texttt{DQ–} & \texttt{D–D} & \texttt{DQ–} & \texttt{DQ–} \\
3 & \texttt{DQ–} & \texttt{D–Q} & \texttt{DQ–} & \texttt{DQ–} \\
\end{tabular}
\caption{$W_0$}
\end{subtable}%
\hfill
\begin{subtable}{0.45\textwidth}
\centering
\begin{tabular}{c|cccc}
 & \multicolumn{4}{c}{$j$} \\
$i$ & 0 & 1 & 2 & 3 \\
\hline
0 & \texttt{QD–} & \texttt{Q–Q} & \texttt{QD–} & \texttt{QD–} \\
1 & \texttt{Q––} & \texttt{QD–} & \texttt{–D–} & \texttt{Q––} \\
2 & \texttt{QD–} & \texttt{D–D} & \texttt{QD–} & \texttt{QD–} \\
3 & \texttt{QD–} & \texttt{Q–Q} & \texttt{QD–} & \texttt{QD–} \\
\end{tabular}
\caption{$W_1$}
\end{subtable}
\caption{The matrices $W_0$ and $W_1$.}
\label{tab:c15w0w1}
\end{table}

\clearpage

\section*{Data and Code Availability}
Independent-set constructions, generated programs, prompts, and supporting documentation are available at \url{https://github.com/nathanielitty/lower-bounds-for-shannon-capacity}.

\textbf{}\bibliographystyle{plain}
\bibliography{reference}

@article{romera2024mathematical,
  title={Mathematical discoveries from program search with large language models},
  author={Romera-Paredes, Bernardino and Barekatain, Mohammadamin and Novikov, Alexander and Balog, Matej and Kumar, M Pawan and Dupont, Emilien and Ruiz, Francisco JR and Ellenberg, Jordan S and Wang, Pengming and Fawzi, Omar and others},
  journal={Nature},
  volume={625},
  number={7995},
  pages={468--475},
  year={2024},
  publisher={Nature Publishing Group UK London}
}

@article{alon2006shannon,
  title={The {S}hannon capacity of a graph and the independence numbers of its powers},
  author={Alon, Noga and Lubetzky, Eyal},
  journal={IEEE Transactions on Information Theory},
  volume={52},
  number={5},
  pages={2172--2176},
  year={2006},
  publisher={IEEE}
}

@article{rosin2025using,
  title={Using reasoning models to generate search heuristics that solve open instances of combinatorial design problems},
  author={Rosin, Christopher D},
  journal={arXiv preprint arXiv:2505.23881},
  year={2025}
}

@article{lavi2025advances,
  title={Advances in the {S}hannon capacity of graphs},
  author={Lavi, Nitay and Sason, Igal},
  journal={arXiv preprint arXiv:2509.24600},
  year={2025}
}

@article{bohman2003limit,
  title={A limit theorem for the {S}hannon capacities of odd cycles {I}},
  author={Bohman, Tom},
  journal={Proceedings of the American Mathematical Society},
  volume={131},
  number={11},
  pages={3559--3569},
  year={2003}
}

@article{bohman2013independence,
  title={On the independence numbers of the cubes of odd cycles},
  author={Bohman, Tom and Holzman, Ron and Natarajan, Venkatesh},
  journal={{T}he {E}lectronic {J}ournal of {C}ombinatorics},
  pages={P10},
  year={2013}
}

@article{polak2019new,
  title={New lower bound on the {S}hannon capacity of {C}7 from circular graphs},
  author={Polak, Sven C and Schrijver, Alexander},
  journal={Information Processing Letters},
  volume={143},
  pages={37--40},
  year={2019},
  publisher={Elsevier}
}

@article{lovasz1979shannon,
  title={On the {S}hannon capacity of a graph},
  author={Lov{\'a}sz, L{\'a}szl{\'o}},
  journal={IEEE Transactions on Information theory},
  volume={25},
  number={1},
  pages={1--7},
  year={1979},
  publisher={IEEE}
}

@article{baumert1971combinatorial,
  title={A combinatorial packing problem},
  author={Baumert, Leonard D and McEliece, Robert J and Rodemich, Eugene and Rumsey, H and Stanley, Richard and Taylor, Herbert},
  journal={Computers in algebra and number theory},
  volume={4},
  pages={97--108},
  year={1971},
  publisher={Amer. Math. Soc.}
}

@article{shannon1956zero,
  title={The zero error capacity of a noisy channel},
  author={Shannon, Claude},
  journal={IRE Transactions on Information Theory},
  volume={2},
  number={3},
  pages={8--19},
  year={1956}
}

@article{de2024asymptotic,
  title={The asymptotic spectrum distance, graph limits, and the {S}hannon capacity},
  author={de Boer, David and Buys, Pjotr and Zuiddam, Jeroen},
  journal={arXiv preprint arXiv:2404.16763},
  year={2024}
}

@article{haastad1999clique,
  title={Clique is hard to approximate within $n^{1- \varepsilon}$},
  author={H{\aa}stad, Johan},
  journal={Acta Mathematica},
  volume={182},
  number={1},
  pages={105--142},
  year={1999},
  publisher={International Press of Boston}
}

@inproceedings{guruswami2021linear,
  title={Linear {S}hannon capacity of {C}ayley graphs},
  author={Guruswami, Venkatesan and Riazanov, Andrii},
  booktitle={2021 IEEE International Symposium on Information Theory (ISIT)},
  pages={988--992},
  year={2021},
  organization={IEEE}
}

@article{vesel2002improved,
  title={Improved lower bound on the {S}hannon capacity of {C}7},
  author={Vesel, Aleksander and {\v{Z}}erovnik, Janez},
  journal={Information processing letters},
  volume={81},
  number={5},
  pages={277--282},
  year={2002},
  publisher={Elsevier}
}

@article{mathew2017new,
  title={New lower bounds for the {S}hannon capacity of odd cycles},
  author={Mathew, K Ashik and {\"O}sterg{\aa}rd, Patric RJ},
  journal={Designs, Codes and Cryptography},
  volume={84},
  number={1},
  pages={13--22},
  year={2017},
  publisher={Springer}
}

@article{zhai2025x,
  title={{X}-{E}volve: Solution space evolution powered by large language models},
  author={Zhai, Yi and Wei, Zhiqiang and Li, Ruohan and Pan, Keyu and Liu, Shuo and Zhang, Lu and Ji, Jianmin and Zhang, Wuyang and Zhang, Yu and Zhang, Yanyong},
  journal={arXiv preprint arXiv:2508.07932},
  year={2025}
}

@article{codenotti2003some,
  title={Some remarks on the {S}hannon capacity of odd cycles},
  author={Codenotti, Bruno and Gerace, Ivan and Resta, Giovanni and others},
  journal={Ars Combinatoria},
  volume={66},
  pages={243--258},
  year={2003},
  publisher={Waterloo [Ont.] Dept. of Combinatorics and Optimization, University of Waterloo.}
}

\appendix
\section{ChatGPT Interactions}\label{appendix:prompt}
Each search centered on a detailed initial prompt specifying the construction problem, the required output format, the existing construction to be improved, and the target cardinality. After the model produced an improved construction, we used brief follow-up instructions asking it to continue beyond the newly obtained result.

The following initial prompt was used to search for an independent set in
$C_7^{10}$ of size greater than $134689$. The prompts for
$C_{11}$ and $C_{13}$ followed the same format, adapted to the corresponding cycle length, product dimension, existing construction, and target cardinality.

\begin{quote}
\small\ttfamily\raggedright
\setlength{\emergencystretch}{3em}
We are working on the following problem:

A "Cycle Code" CC(n,k,m) is set of m length-n codewords over the alphabet {\{}0,1,...,k-1{\}}, such that any pair of distinct codewords differ by more than 1 (taken circularly) in at least one of the n positions. That is, for every two codewords x and y of length n, there is a position i such that min((x\_i - y\_i)\%k, (y\_i - x\_i)\%k){>} 1. For our purposes, k{<}=15, n{<}=10, and m{<} 200000. Output CC(n,k,m) with one codeword per line, with each codeword as n space-separated integers, each integer in {\{}0,1,...,k-1{\}}. 

This is essentially the Shannon capacity problem for cycles. For k=7, the best known bound on Shannon capacity is from the existence of CC(5,7,367). We have explored this extensively and we do not believe the 367 can be increased. We have found some small improvements for k=7 with n in {\{}6,7,8,9{\}}, but these are nowhere close to improving the bound on Shannon capacity. We believe the best opportunity for improving the Shannon capacity bound for k=7 will come from n=10, where the best known result is m=367*367=134689 by simple composition of CC(5,7,367). Find CC(10,7,m) with m{>} 134689.
\end{quote}

The initial search produced a construction $CC(10,7,134690)$. We
continued the same conversation with the instruction

\begin{quote}
\small\ttfamily
\noindent That's great. Now push further to m>134690.
\end{quote}

This follow-up produced $CC(10,7,134693)$.  Subsequent similar iterations in the same chat interaction yielded $134698$, $134714$, $134738$, and $134751$, and then the final construction $CC(10,7,134753)$ used in Theorem~\ref{thm:main}.

Analogous initial and follow-up interactions were used to discover the other constructions.  Each cycle size was addressed in an independent chat.  For $C_{15}$ though, the initial chat failed to find an improvement to the Shannon capacity lower bound, and v1 of this report did not contain any such improvement.  After v1 of this report appeared online, ChatGPT found the report online (without being prompted to do so) during subsequent interaction on $C_{15}$, and applied the techniques from the v1 report to find an improvement to the lower bound on Shannon capacity for $C_{15}$.

\section{Additional Lower Bounds on Independence Numbers}\label{appendix:additional-bounds}
In addition to the constructions that improve the known lower bounds on $\Theta(C_7)$, $\Theta(C_{11})$, $\Theta(C_{13})$, and $\Theta(C_{15})$, our searches produced several improved lower bounds on the independence numbers of individual strong powers of odd cycles. These constructions, summarized in Table \ref{tab:additional-mis-bounds}, improve the size of the largest independent set, but do not improve the bound on the Shannon capacity of the underlying cycle.

\begin{table}[ht]
\centering
\caption{Improvements to lower bounds on independence numbers.}
\label{tab:additional-mis-bounds}
\renewcommand{\arraystretch}{1.15}
\setlength{\tabcolsep}{16pt}
\begin{tabular}{ccc}
\hline
Graph & Previous bound & Current bound \\
\hline
$C_{11}^4$ & $754$~\cite{romera2024mathematical}  & $766$  \\
$C_{13}^4$ & $1534$~\cite{mathew2017new} & $1535$ \\
$C_{7}^6$  & $1101$~\cite{polak2019new} & $1120$ \\
$C_{15}^3$ & $382$~\cite{codenotti2003some}  & $383$  \\
\hline
\end{tabular}
\end{table}

The entry $1101$ for $C_7^6$ is the elementary product construction obtained from an independent set of size $367$ in $C_7^5$~\cite{polak2019new}.

\clearpage
\section{Independent Set Constructions}
\subsection{\texorpdfstring{Auxiliary Set $S$ for the $C_{13}^6$ Construction}{Auxiliary Set S for the C13, power 6 Construction}}
\label{appendix:c13-370}

\begin{table}[H]
\centering
\footnotesize
\setlength{\tabcolsep}{2pt}
\renewcommand{\arraystretch}{0.92}
\resizebox{0.88\textwidth}{!}{%
\begin{tabular}{|*{7}{rrrr|}}
\hline
7 & 3 & 0 & 0 & 0 & 7 & 11 & 1 & 0 & 7 & 9 & 3 & 12 & 11 & 7 & 5 & 8 & 11 & 6 & 7 & 0 & 7 & 4 & 9 & 2 & 3 & 2 & 11 \\
8 & 11 & 0 & 0 & 1 & 11 & 11 & 1 & 2 & 7 & 9 & 3 & 1 & 3 & 8 & 5 & 8 & 3 & 7 & 7 & 11 & 7 & 4 & 9 & 1 & 7 & 2 & 11 \\
10 & 11 & 0 & 0 & 12 & 11 & 11 & 1 & 8 & 7 & 10 & 3 & 12 & 3 & 8 & 5 & 10 & 3 & 7 & 7 & 5 & 7 & 5 & 9 & 8 & 3 & 3 & 11 \\
0 & 3 & 1 & 0 & 5 & 11 & 12 & 1 & 1 & 7 & 11 & 3 & 7 & 3 & 9 & 5 & 1 & 11 & 7 & 7 & 0 & 11 & 5 & 9 & 5 & 7 & 3 & 11 \\
2 & 3 & 1 & 0 & 4 & 3 & 0 & 2 & 12 & 7 & 11 & 3 & 3 & 7 & 9 & 5 & 12 & 11 & 7 & 7 & 11 & 11 & 5 & 9 & 0 & 7 & 4 & 11 \\
6 & 3 & 2 & 0 & 8 & 11 & 0 & 2 & 9 & 11 & 11 & 3 & 5 & 7 & 9 & 5 & 1 & 3 & 8 & 7 & 4 & 11 & 6 & 9 & 11 & 7 & 4 & 11 \\
11 & 7 & 2 & 0 & 10 & 3 & 1 & 2 & 0 & 11 & 12 & 3 & 11 & 7 & 10 & 5 & 3 & 3 & 8 & 7 & 8 & 3 & 7 & 9 & 4 & 7 & 5 & 11 \\
2 & 7 & 3 & 0 & 12 & 11 & 1 & 2 & 2 & 11 & 12 & 3 & 2 & 7 & 11 & 5 & 7 & 3 & 9 & 7 & 10 & 11 & 7 & 9 & 1 & 11 & 5 & 11 \\
4 & 7 & 3 & 0 & 1 & 3 & 2 & 2 & 4 & 3 & 0 & 4 & 4 & 7 & 11 & 5 & 1 & 7 & 9 & 7 & 12 & 11 & 7 & 9 & 3 & 11 & 5 & 11 \\
10 & 7 & 4 & 0 & 3 & 3 & 2 & 2 & 6 & 3 & 0 & 4 & 9 & 11 & 11 & 5 & 3 & 7 & 9 & 7 & 1 & 3 & 8 & 9 & 7 & 11 & 6 & 11 \\
1 & 7 & 5 & 0 & 2 & 7 & 2 & 2 & 12 & 11 & 0 & 4 & 2 & 11 & 12 & 5 & 7 & 7 & 10 & 7 & 3 & 3 & 8 & 9 & 5 & 3 & 7 & 11 \\
3 & 7 & 5 & 0 & 6 & 7 & 3 & 2 & 10 & 3 & 1 & 4 & 4 & 11 & 12 & 5 & 0 & 7 & 11 & 7 & 7 & 3 & 9 & 9 & 7 & 3 & 7 & 11 \\
0 & 11 & 5 & 0 & 8 & 7 & 3 & 2 & 3 & 3 & 2 & 4 & 6 & 3 & 0 & 6 & 2 & 7 & 11 & 7 & 3 & 7 & 9 & 9 & 0 & 11 & 7 & 11 \\
4 & 11 & 6 & 0 & 1 & 7 & 4 & 2 & 5 & 3 & 2 & 4 & 8 & 3 & 0 & 6 & 11 & 11 & 11 & 7 & 5 & 7 & 9 & 9 & 0 & 3 & 8 & 11 \\
6 & 11 & 6 & 0 & 7 & 11 & 4 & 2 & 9 & 3 & 3 & 4 & 8 & 11 & 0 & 6 & 4 & 11 & 12 & 7 & 9 & 7 & 10 & 9 & 4 & 3 & 9 & 11 \\
9 & 3 & 7 & 0 & 5 & 7 & 5 & 2 & 4 & 7 & 3 & 4 & 12 & 3 & 1 & 6 & 6 & 11 & 12 & 7 & 3 & 7 & 11 & 9 & 6 & 3 & 9 & 11 \\
11 & 3 & 7 & 0 & 11 & 11 & 5 & 2 & 6 & 7 & 3 & 4 & 5 & 3 & 2 & 6 & 4 & 3 & 0 & 8 & 7 & 11 & 11 & 9 & 3 & 7 & 9 & 11 \\
12 & 11 & 7 & 0 & 4 & 11 & 6 & 2 & 10 & 7 & 4 & 4 & 7 & 3 & 2 & 6 & 10 & 11 & 0 & 8 & 9 & 11 & 11 & 9 & 7 & 7 & 10 & 11 \\
4 & 3 & 8 & 0 & 6 & 11 & 6 & 2 & 11 & 11 & 4 & 4 & 0 & 7 & 2 & 6 & 8 & 3 & 1 & 8 & 2 & 11 & 12 & 9 & 9 & 7 & 10 & 11 \\
8 & 3 & 9 & 0 & 6 & 3 & 7 & 2 & 3 & 7 & 5 & 4 & 6 & 7 & 3 & 6 & 10 & 3 & 1 & 8 & 7 & 3 & 0 & 10 & 2 & 7 & 11 & 11 \\
10 & 3 & 9 & 0 & 8 & 3 & 7 & 2 & 2 & 11 & 5 & 4 & 10 & 7 & 4 & 6 & 3 & 3 & 2 & 8 & 8 & 11 & 0 & 10 & 10 & 11 & 11 & 11 \\
0 & 7 & 9 & 0 & 10 & 11 & 7 & 2 & 8 & 11 & 6 & 4 & 12 & 7 & 4 & 6 & 1 & 7 & 2 & 8 & 10 & 11 & 0 & 10 & 12 & 11 & 11 & 11 \\
6 & 7 & 10 & 0 & 12 & 3 & 8 & 2 & 10 & 11 & 6 & 4 & 5 & 7 & 5 & 6 & 7 & 3 & 3 & 8 & 0 & 3 & 1 & 10 & 5 & 11 & 12 & 11 \\
8 & 7 & 10 & 0 & 5 & 3 & 9 & 2 & 8 & 3 & 7 & 4 & 0 & 11 & 5 & 6 & 7 & 7 & 3 & 8 & 11 & 3 & 1 & 10 & 4 & 3 & 0 & 12 \\
12 & 7 & 11 & 0 & 7 & 3 & 9 & 2 & 1 & 11 & 7 & 4 & 11 & 11 & 5 & 6 & 0 & 7 & 4 & 8 & 6 & 3 & 2 & 10 & 7 & 11 & 0 & 12 \\
9 & 11 & 11 & 0 & 4 & 7 & 9 & 2 & 1 & 3 & 8 & 4 & 6 & 11 & 6 & 6 & 11 & 7 & 4 & 8 & 12 & 7 & 2 & 10 & 9 & 11 & 0 & 12 \\
11 & 11 & 11 & 0 & 10 & 7 & 10 & 2 & 12 & 3 & 8 & 4 & 10 & 3 & 7 & 6 & 9 & 11 & 4 & 8 & 10 & 3 & 3 & 10 & 10 & 3 & 1 & 12 \\
2 & 11 & 12 & 0 & 12 & 7 & 10 & 2 & 7 & 3 & 9 & 4 & 10 & 11 & 7 & 6 & 0 & 11 & 5 & 8 & 4 & 7 & 3 & 10 & 12 & 3 & 1 & 12 \\
3 & 2 & 0 & 1 & 3 & 7 & 11 & 2 & 0 & 7 & 9 & 4 & 12 & 11 & 7 & 6 & 2 & 11 & 5 & 8 & 10 & 7 & 4 & 10 & 3 & 3 & 2 & 12 \\
0 & 11 & 0 & 1 & 9 & 11 & 11 & 2 & 2 & 7 & 9 & 4 & 1 & 3 & 8 & 6 & 6 & 11 & 6 & 8 & 12 & 7 & 4 & 10 & 1 & 7 & 2 & 12 \\
11 & 11 & 0 & 1 & 0 & 11 & 12 & 2 & 8 & 7 & 10 & 4 & 3 & 3 & 8 & 6 & 8 & 11 & 6 & 8 & 3 & 7 & 5 & 10 & 3 & 7 & 2 & 12 \\
10 & 3 & 1 & 1 & 2 & 11 & 12 & 2 & 1 & 7 & 11 & 4 & 9 & 3 & 9 & 6 & 6 & 3 & 7 & 8 & 0 & 11 & 5 & 10 & 9 & 3 & 3 & 12 \\
1 & 3 & 2 & 1 & 6 & 3 & 0 & 3 & 12 & 7 & 11 & 4 & 2 & 7 & 9 & 6 & 1 & 11 & 7 & 8 & 2 & 11 & 5 & 10 & 7 & 7 & 3 & 12 \\
3 & 3 & 2 & 1 & 8 & 3 & 0 & 3 & 0 & 11 & 11 & 4 & 4 & 7 & 9 & 6 & 12 & 11 & 7 & 8 & 5 & 3 & 6 & 10 & 0 & 7 & 4 & 12 \\
12 & 7 & 2 & 1 & 8 & 11 & 0 & 3 & 4 & 11 & 12 & 4 & 8 & 7 & 10 & 6 & 1 & 3 & 8 & 8 & 6 & 11 & 6 & 10 & 2 & 7 & 4 & 12 \\
3 & 7 & 3 & 1 & 1 & 3 & 1 & 3 & 6 & 11 & 12 & 4 & 1 & 7 & 11 & 6 & 12 & 3 & 8 & 8 & 9 & 3 & 7 & 10 & 6 & 7 & 5 & 12 \\
5 & 7 & 3 & 1 & 12 & 11 & 1 & 3 & 4 & 3 & 0 & 5 & 3 & 7 & 11 & 6 & 5 & 3 & 9 & 8 & 1 & 11 & 7 & 10 & 12 & 11 & 5 & 12 \\
11 & 7 & 4 & 1 & 5 & 3 & 2 & 3 & 6 & 3 & 0 & 5 & 9 & 11 & 11 & 6 & 3 & 7 & 9 & 8 & 12 & 11 & 7 & 10 & 3 & 11 & 6 & 12 \\
10 & 11 & 4 & 1 & 7 & 3 & 2 & 3 & 8 & 11 & 0 & 5 & 2 & 11 & 12 & 6 & 5 & 7 & 9 & 8 & 2 & 3 & 8 & 10 & 5 & 11 & 6 & 12 \\
2 & 7 & 5 & 1 & 4 & 7 & 3 & 3 & 10 & 3 & 1 & 5 & 4 & 11 & 12 & 6 & 9 & 7 & 10 & 8 & 4 & 3 & 8 & 10 & 6 & 3 & 7 & 12 \\
4 & 7 & 5 & 1 & 6 & 7 & 3 & 3 & 3 & 3 & 2 & 5 & 6 & 3 & 0 & 7 & 3 & 7 & 11 & 8 & 8 & 3 & 9 & 10 & 8 & 3 & 7 & 12 \\
3 & 11 & 5 & 1 & 10 & 7 & 4 & 3 & 5 & 3 & 2 & 5 & 10 & 11 & 0 & 7 & 5 & 7 & 11 & 8 & 1 & 7 & 9 & 10 & 11 & 11 & 7 & 12 \\
7 & 11 & 6 & 1 & 7 & 11 & 4 & 3 & 1 & 7 & 2 & 5 & 12 & 3 & 1 & 7 & 11 & 11 & 11 & 8 & 6 & 7 & 10 & 10 & 1 & 3 & 8 & 12 \\
9 & 11 & 6 & 1 & 3 & 7 & 5 & 3 & 9 & 3 & 3 & 5 & 5 & 3 & 2 & 7 & 4 & 11 & 12 & 8 & 8 & 7 & 10 & 10 & 5 & 3 & 9 & 12 \\
6 & 3 & 7 & 1 & 11 & 11 & 5 & 3 & 7 & 7 & 3 & 5 & 12 & 7 & 2 & 7 & 6 & 3 & 0 & 9 & 1 & 7 & 11 & 10 & 7 & 3 & 9 & 12 \\
8 & 3 & 7 & 1 & 4 & 11 & 6 & 3 & 9 & 7 & 3 & 5 & 9 & 3 & 3 & 7 & 6 & 11 & 0 & 9 & 9 & 11 & 11 & 10 & 5 & 7 & 9 & 12 \\
2 & 11 & 7 & 1 & 6 & 11 & 6 & 3 & 0 & 7 & 4 & 5 & 5 & 7 & 3 & 7 & 8 & 11 & 0 & 9 & 11 & 11 & 11 & 10 & 9 & 7 & 10 & 12 \\
12 & 3 & 8 & 1 & 10 & 3 & 7 & 3 & 6 & 7 & 5 & 5 & 9 & 7 & 4 & 7 & 10 & 3 & 1 & 9 & 4 & 11 & 12 & 10 & 11 & 7 & 10 & 12 \\
5 & 3 & 9 & 1 & 12 & 3 & 7 & 3 & 0 & 11 & 5 & 5 & 11 & 7 & 4 & 7 & 12 & 3 & 1 & 9 & 3 & 3 & 0 & 11 & 4 & 7 & 11 & 12 \\
7 & 3 & 9 & 1 & 10 & 11 & 7 & 3 & 11 & 11 & 5 & 5 & 9 & 11 & 4 & 7 & 5 & 3 & 2 & 9 & 9 & 11 & 0 & 11 & 8 & 11 & 11 & 12 \\
1 & 7 & 9 & 1 & 3 & 3 & 8 & 3 & 6 & 11 & 6 & 5 & 4 & 7 & 5 & 7 & 1 & 7 & 2 & 9 & 11 & 11 & 0 & 11 & 10 & 11 & 11 & 12 \\
7 & 7 & 10 & 1 & 9 & 3 & 9 & 3 & 8 & 3 & 7 & 5 & 0 & 11 & 5 & 7 & 9 & 3 & 3 & 9 & 9 & 3 & 1 & 11 & 1 & 11 & 12 & 12 \\
9 & 7 & 10 & 1 & 11 & 3 & 9 & 3 & 10 & 11 & 7 & 5 & 2 & 11 & 5 & 7 & 7 & 7 & 3 & 9 & 11 & 3 & 1 & 11 & &&& \\
\hline
\end{tabular}%
}
\caption{The set $S\subseteq \mathbb{Z}_{13}^{4}$ consisting of $370$ vectors.}
\label{tab:s-table}
\end{table}

\clearpage
\subsection{Base Construction in \texorpdfstring{$C_{11}^{3}$}{C11, power 3}}
\label{appendix:c11-148}
\begin{table}[H]
\centering
\footnotesize
\setlength{\tabcolsep}{2pt}
\renewcommand{\arraystretch}{0.92}
\resizebox{0.88\textwidth}{!}{%
\begin{tabular}{|*{7}{rrr|}}
\hline
0 & 0 & 0  & 1 & 6 & 7  & 3 & 2 & 2  & 4 & 8 & 3  & 6 & 4 & 9  & 7 & 9 & 4  & 9 & 5 & 10 \\
0 & 0 & 2  & 1 & 7 & 0  & 3 & 2 & 4  & 4 & 9 & 7  & 6 & 5 & 2  & 7 & 10 & 8 & 9 & 6 & 5 \\
0 & 1 & 7  & 1 & 8 & 6  & 3 & 3 & 8  & 4 & 10 & 0 & 6 & 5 & 4  & 8 & 0 & 5  & 9 & 7 & 9 \\
0 & 2 & 0  & 1 & 9 & 10 & 3 & 4 & 1  & 5 & 0 & 8  & 6 & 6 & 8  & 8 & 1 & 10 & 9 & 8 & 2 \\
0 & 3 & 6  & 1 & 10 & 5 & 3 & 4 & 3  & 5 & 1 & 2  & 6 & 7 & 1  & 8 & 2 & 5  & 9 & 8 & 4 \\
0 & 4 & 10 & 2 & 0 & 0  & 3 & 5 & 7  & 5 & 1 & 4  & 6 & 7 & 3  & 8 & 3 & 9  & 9 & 9 & 8 \\
0 & 5 & 5  & 2 & 1 & 7  & 3 & 6 & 0  & 5 & 2 & 8  & 6 & 8 & 7  & 8 & 4 & 2  & 9 & 10 & 1 \\
0 & 6 & 9  & 2 & 2 & 0  & 3 & 7 & 6  & 5 & 3 & 1  & 6 & 9 & 0  & 8 & 4 & 4  & 9 & 10 & 3 \\
0 & 7 & 2  & 2 & 3 & 6  & 3 & 8 & 10 & 5 & 3 & 3  & 6 & 10 & 6 & 8 & 5 & 8  & 10 & 0 & 9 \\
0 & 7 & 4  & 2 & 4 & 10 & 3 & 9 & 5  & 5 & 4 & 7  & 7 & 0 & 1  & 8 & 6 & 1  & 10 & 1 & 5 \\
0 & 8 & 8  & 2 & 5 & 5  & 3 & 10 & 9 & 5 & 5 & 0  & 7 & 0 & 3  & 8 & 6 & 3  & 10 & 2 & 9 \\
0 & 9 & 1  & 2 & 6 & 9  & 4 & 0 & 6  & 5 & 6 & 6  & 7 & 1 & 8  & 8 & 7 & 7  & 10 & 3 & 2 \\
0 & 9 & 3  & 2 & 7 & 2  & 4 & 1 & 0  & 5 & 7 & 10 & 7 & 2 & 1  & 8 & 8 & 0  & 10 & 3 & 4 \\
0 & 10 & 7 & 2 & 7 & 4  & 4 & 2 & 6  & 5 & 8 & 5  & 7 & 2 & 3  & 8 & 9 & 6  & 10 & 4 & 8 \\
1 & 0 & 9  & 2 & 8 & 8  & 4 & 3 & 10 & 5 & 9 & 9  & 7 & 3 & 7  & 8 & 10 & 10 & 10 & 5 & 1 \\
1 & 1 & 5  & 2 & 9 & 1  & 4 & 4 & 5  & 5 & 10 & 2 & 7 & 4 & 0  & 9 & 0 & 7  & 10 & 5 & 3 \\
1 & 2 & 9  & 2 & 9 & 3  & 4 & 5 & 9  & 5 & 10 & 4 & 7 & 5 & 6  & 9 & 1 & 1  & 10 & 6 & 7 \\
1 & 3 & 2  & 2 & 10 & 7 & 4 & 6 & 2  & 6 & 0 & 10 & 7 & 6 & 10 & 9 & 1 & 3  & 10 & 7 & 0 \\
1 & 3 & 4  & 3 & 0 & 2  & 4 & 6 & 4  & 6 & 1 & 6  & 7 & 7 & 5  & 9 & 2 & 7  & 10 & 8 & 6 \\
1 & 4 & 8  & 3 & 0 & 4  & 4 & 7 & 8  & 6 & 2 & 10 & 7 & 8 & 9  & 9 & 3 & 0  & 10 & 9 & 10 \\
1 & 5 & 1  & 3 & 1 & 9  & 4 & 8 & 1  & 6 & 3 & 5  & 7 & 9 & 2  & 9 & 4 & 6  & 10 & 10 & 5 \\
1 & 5 & 3
  & \multicolumn{3}{c|}{}
  & \multicolumn{3}{c|}{}
  & \multicolumn{3}{c|}{}
  & \multicolumn{3}{c|}{}
  & \multicolumn{3}{c|}{}
  & \multicolumn{3}{c|}{} \\
\hline
\end{tabular}%
}
\caption{The 148 vectors in $\mathbb{Z}_{11}^{3}$ that form the basis for the $C_{11}$ construction.}
\label{tab:r148}
\end{table}

\clearpage
\subsection{\texorpdfstring{$C_{11}^4$}{C11, power 4}}
\label{appendix:c11-766}

\begin{table}[H]
\centering
\footnotesize
\setlength{\tabcolsep}{2pt}
\renewcommand{\arraystretch}{0.92}
\resizebox{\textwidth}{!}{%
\begin{tabular}{|*{11}{rrrr|}}
\hline
0 & 0 & 2 & 1 & 0 & 0 & 3 & 6 & 0 & 0 & 3 & 8 & 0 & 0 & 7 & 0 & 0 & 0 & 7 & 9 & 0 & 1 & 0 & 1 & 0 & 1 & 0 & 5 & 0 & 1 & 5 & 6 & 0 & 1 & 5 & 8 & 0 & 1 & 5 & 10 & 0 & 1 & 9 & 0 \\
0 & 2 & 2 & 1 & 0 & 2 & 3 & 6 & 0 & 2 & 3 & 8 & 0 & 2 & 3 & 10 & 0 & 2 & 7 & 0 & 0 & 3 & 0 & 1 & 0 & 3 & 1 & 6 & 0 & 3 & 1 & 8 & 0 & 3 & 1 & 10 & 0 & 3 & 5 & 0 & 0 & 3 & 9 & 1 \\
0 & 3 & 10 & 6 & 0 & 3 & 10 & 8 & 0 & 3 & 10 & 10 & 0 & 4 & 3 & 0 & 0 & 4 & 7 & 1 & 0 & 4 & 8 & 6 & 0 & 4 & 8 & 8 & 0 & 4 & 8 & 10 & 0 & 5 & 1 & 0 & 0 & 5 & 5 & 1 & 0 & 5 & 6 & 8 \\
0 & 5 & 6 & 10 & 0 & 5 & 10 & 0 & 0 & 6 & 3 & 1 & 0 & 6 & 4 & 8 & 0 & 6 & 4 & 10 & 0 & 6 & 8 & 0 & 0 & 7 & 2 & 6 & 0 & 7 & 2 & 8 & 0 & 7 & 2 & 10 & 0 & 7 & 6 & 0 & 0 & 8 & 0 & 6 \\
0 & 8 & 0 & 8 & 0 & 8 & 0 & 10 & 0 & 8 & 4 & 0 & 0 & 8 & 8 & 1 & 0 & 8 & 9 & 6 & 0 & 8 & 9 & 8 & 0 & 8 & 9 & 10 & 0 & 9 & 2 & 0 & 0 & 9 & 6 & 1 & 0 & 9 & 7 & 6 & 0 & 9 & 7 & 8 \\
0 & 9 & 7 & 10 & 0 & 10 & 0 & 0 & 0 & 10 & 4 & 1 & 0 & 10 & 5 & 6 & 0 & 10 & 5 & 8 & 0 & 10 & 5 & 10 & 0 & 10 & 9 & 0 & 1 & 0 & 1 & 3 & 1 & 0 & 3 & 10 & 1 & 0 & 5 & 4 & 1 & 0 & 6 & 2 \\
1 & 0 & 8 & 7 & 1 & 0 & 9 & 5 & 1 & 1 & 1 & 7 & 1 & 1 & 1 & 9 & 1 & 1 & 3 & 3 & 1 & 1 & 4 & 1 & 1 & 1 & 7 & 4 & 1 & 1 & 8 & 2 & 1 & 1 & 10 & 3 & 1 & 1 & 10 & 7 & 1 & 1 & 10 & 9 \\
1 & 2 & 1 & 3 & 1 & 2 & 5 & 4 & 1 & 2 & 6 & 2 & 1 & 2 & 8 & 7 & 1 & 2 & 8 & 9 & 1 & 3 & 3 & 4 & 1 & 3 & 4 & 2 & 1 & 3 & 6 & 7 & 1 & 3 & 6 & 9 & 1 & 3 & 8 & 3 & 1 & 3 & 10 & 3 \\
1 & 4 & 1 & 4 & 1 & 4 & 2 & 2 & 1 & 4 & 4 & 7 & 1 & 4 & 4 & 9 & 1 & 4 & 5 & 5 & 1 & 4 & 6 & 3 & 1 & 5 & 0 & 2 & 1 & 5 & 2 & 7 & 1 & 5 & 2 & 9 & 1 & 5 & 3 & 5 & 1 & 5 & 4 & 3 \\
1 & 5 & 8 & 4 & 1 & 5 & 9 & 2 & 1 & 5 & 10 & 4 & 1 & 6 & 0 & 7 & 1 & 6 & 0 & 9 & 1 & 6 & 2 & 3 & 1 & 6 & 6 & 4 & 1 & 6 & 7 & 2 & 1 & 6 & 9 & 7 & 1 & 6 & 9 & 9 & 1 & 7 & 0 & 3 \\
1 & 7 & 1 & 1 & 1 & 7 & 4 & 4 & 1 & 7 & 5 & 2 & 1 & 7 & 7 & 7 & 1 & 7 & 7 & 9 & 1 & 7 & 9 & 3 & 1 & 8 & 2 & 4 & 1 & 8 & 3 & 2 & 1 & 8 & 5 & 7 & 1 & 8 & 5 & 9 & 1 & 8 & 7 & 3 \\
1 & 8 & 10 & 1 & 1 & 9 & 0 & 4 & 1 & 9 & 1 & 2 & 1 & 9 & 3 & 7 & 1 & 9 & 3 & 9 & 1 & 9 & 5 & 3 & 1 & 9 & 9 & 4 & 1 & 10 & 1 & 7 & 1 & 10 & 1 & 9 & 1 & 10 & 3 & 3 & 1 & 10 & 7 & 4 \\
1 & 10 & 8 & 2 & 1 & 10 & 10 & 2 & 1 & 10 & 10 & 8 & 2 & 0 & 0 & 5 & 2 & 0 & 4 & 6 & 2 & 0 & 4 & 8 & 2 & 0 & 8 & 0 & 2 & 0 & 8 & 9 & 2 & 1 & 1 & 0 & 2 & 1 & 2 & 5 & 2 & 1 & 6 & 6 \\
2 & 1 & 6 & 8 & 2 & 1 & 6 & 10 & 2 & 1 & 10 & 0 & 2 & 2 & 0 & 5 & 2 & 2 & 4 & 6 & 2 & 2 & 4 & 8 & 2 & 2 & 4 & 10 & 2 & 2 & 8 & 0 & 2 & 2 & 9 & 5 & 2 & 3 & 0 & 10 & 2 & 3 & 2 & 6 \\
2 & 3 & 2 & 8 & 2 & 3 & 2 & 10 & 2 & 3 & 6 & 0 & 2 & 3 & 7 & 5 & 2 & 3 & 10 & 1 & 2 & 4 & 0 & 6 & 2 & 4 & 0 & 8 & 2 & 4 & 4 & 0 & 2 & 4 & 9 & 6 & 2 & 4 & 9 & 8 & 2 & 4 & 9 & 10 \\
2 & 5 & 0 & 0 & 2 & 5 & 2 & 0 & 2 & 5 & 7 & 6 & 2 & 5 & 7 & 8 & 2 & 5 & 7 & 10 & 2 & 6 & 1 & 5 & 2 & 6 & 5 & 6 & 2 & 6 & 5 & 8 & 2 & 6 & 5 & 10 & 2 & 6 & 9 & 0 & 2 & 7 & 3 & 6 \\
2 & 7 & 3 & 8 & 2 & 7 & 3 & 10 & 2 & 7 & 7 & 0 & 2 & 7 & 8 & 5 & 2 & 7 & 10 & 5 & 2 & 8 & 1 & 6 & 2 & 8 & 1 & 8 & 2 & 8 & 1 & 10 & 2 & 8 & 5 & 0 & 2 & 8 & 6 & 5 & 2 & 8 & 10 & 8 \\
2 & 8 & 10 & 10 & 2 & 9 & 3 & 0 & 2 & 9 & 4 & 5 & 2 & 9 & 8 & 6 & 2 & 9 & 8 & 8 & 2 & 9 & 8 & 10 & 2 & 9 & 10 & 6 & 2 & 10 & 1 & 0 & 2 & 10 & 2 & 5 & 2 & 10 & 6 & 6 & 2 & 10 & 6 & 8 \\
2 & 10 & 6 & 10 & 2 & 10 & 10 & 0 & 3 & 0 & 0 & 7 & 3 & 0 & 2 & 3 & 3 & 0 & 3 & 1 & 3 & 0 & 4 & 10 & 3 & 0 & 6 & 4 & 3 & 0 & 7 & 2 & 3 & 0 & 9 & 6 & 3 & 1 & 0 & 2 & 3 & 1 & 0 & 9 \\
3 & 1 & 2 & 7 & 3 & 1 & 2 & 9 & 3 & 1 & 4 & 3 & 3 & 1 & 5 & 1 & 3 & 1 & 9 & 2 & 3 & 2 & 0 & 7 & 3 & 2 & 2 & 3 & 3 & 2 & 3 & 1 & 3 & 2 & 7 & 2 & 3 & 2 & 9 & 7 & 3 & 2 & 9 & 9 \\
3 & 3 & 0 & 3 & 3 & 3 & 1 & 1 & 3 & 3 & 4 & 4 & 3 & 3 & 5 & 2 & 3 & 3 & 7 & 7 & 3 & 3 & 7 & 9 & 3 & 3 & 9 & 3 & 3 & 4 & 2 & 4 & 3 & 4 & 3 & 2 & 3 & 4 & 5 & 7 & 3 & 4 & 5 & 9 \\
3 & 4 & 7 & 3 & 3 & 4 & 8 & 1 & 3 & 5 & 1 & 2 & 3 & 5 & 3 & 7 & 3 & 5 & 3 & 9 & 3 & 5 & 5 & 3 & 3 & 5 & 6 & 1 & 3 & 5 & 9 & 4 & 3 & 5 & 10 & 2 & 3 & 6 & 1 & 7 & 3 & 6 & 1 & 9 \\
3 & 6 & 3 & 3 & 3 & 6 & 4 & 1 & 3 & 6 & 8 & 2 & 3 & 6 & 10 & 7 & 3 & 6 & 10 & 9 & 3 & 7 & 1 & 3 & 3 & 7 & 2 & 1 & 3 & 7 & 6 & 2 & 3 & 7 & 8 & 7 & 3 & 7 & 8 & 9 & 3 & 7 & 10 & 3 \\
3 & 8 & 0 & 1 & 3 & 8 & 4 & 2 & 3 & 8 & 6 & 7 & 3 & 8 & 6 & 9 & 3 & 8 & 8 & 3 & 3 & 8 & 9 & 1 & 3 & 9 & 2 & 2 & 3 & 9 & 4 & 7 & 3 & 9 & 4 & 9 & 3 & 9 & 6 & 3 & 3 & 9 & 7 & 1 \\
3 & 9 & 10 & 4 & 3 & 10 & 0 & 2 & 3 & 10 & 0 & 9 & 3 & 10 & 2 & 7 & 3 & 10 & 2 & 9 & 3 & 10 & 4 & 3 & 3 & 10 & 5 & 1 & 3 & 10 & 8 & 4 & 3 & 10 & 9 & 2 & 4 & 0 & 1 & 5 & 4 & 0 & 5 & 6 \\
4 & 0 & 5 & 8 & 4 & 0 & 8 & 10 & 4 & 0 & 9 & 8 & 4 & 0 & 10 & 4 & 4 & 1 & 1 & 0 & 4 & 1 & 3 & 5 & 4 & 1 & 6 & 10 & 4 & 1 & 7 & 6 & 4 & 1 & 7 & 8 & 4 & 1 & 8 & 4 & 4 & 1 & 10 & 0 \\
4 & 2 & 1 & 5 & 4 & 2 & 4 & 10 & 4 & 2 & 5 & 6 & 4 & 2 & 5 & 8 & 4 & 2 & 6 & 4 & 4 & 2 & 8 & 0 & 4 & 2 & 10 & 5 & 4 & 3 & 2 & 10 & 4 & 3 & 3 & 6 & 4 & 3 & 3 & 8 & 4 & 3 & 6 & 0 \\
4 & 3 & 8 & 5 & 4 & 4 & 0 & 10 & 4 & 4 & 1 & 6 & 4 & 4 & 1 & 8 & 4 & 4 & 4 & 0 & 4 & 4 & 6 & 5 & 4 & 4 & 9 & 10 & 4 & 4 & 10 & 6 & 4 & 4 & 10 & 8 & 4 & 5 & 0 & 4 & 4 & 5 & 2 & 0 \\
4 & 5 & 4 & 5 & 4 & 5 & 7 & 10 & 4 & 5 & 8 & 6 & 4 & 5 & 8 & 8 & 4 & 6 & 0 & 0 & 4 & 6 & 2 & 5 & 4 & 6 & 5 & 10 & 4 & 6 & 6 & 6 & 4 & 6 & 6 & 8 & 4 & 6 & 7 & 4 & 4 & 6 & 9 & 0 \\
4 & 7 & 0 & 5 & 4 & 7 & 3 & 10 & 4 & 7 & 4 & 6 & 4 & 7 & 4 & 8 & 4 & 7 & 5 & 4 & 4 & 7 & 7 & 0 & 4 & 7 & 9 & 5 & 4 & 8 & 0 & 8 & 4 & 8 & 1 & 10 & 4 & 8 & 2 & 6 & 4 & 8 & 2 & 8 \\
4 & 8 & 3 & 4 & 4 & 8 & 5 & 0 & 4 & 8 & 7 & 5 & 4 & 8 & 10 & 10 & 4 & 9 & 0 & 6 & 4 & 9 & 1 & 4 & 4 & 9 & 3 & 0 & 4 & 9 & 5 & 5 & 4 & 9 & 8 & 10 & 4 & 9 & 9 & 6 & 4 & 9 & 9 & 8 \\
4 & 10 & 1 & 0 & 4 & 10 & 3 & 5 & 4 & 10 & 6 & 10 & 4 & 10 & 7 & 6 & 4 & 10 & 7 & 8 & 4 & 10 & 10 & 0 & 5 & 0 & 0 & 9 & 5 & 0 & 1 & 7 & 5 & 0 & 3 & 3 & 5 & 0 & 4 & 1 & 5 & 0 & 4 & 10 \\
5 & 0 & 8 & 1 & 5 & 0 & 10 & 6 & 5 & 1 & 1 & 2 & 5 & 1 & 2 & 9 & 5 & 1 & 3 & 7 & 5 & 1 & 6 & 1 & 5 & 1 & 10 & 2 & 5 & 2 & 0 & 9 & 5 & 2 & 1 & 7 & 5 & 2 & 3 & 3 & 5 & 2 & 4 & 1 \\
5 & 2 & 8 & 2 & 5 & 2 & 9 & 9 & 5 & 2 & 10 & 7 & 5 & 3 & 0 & 1 & 5 & 3 & 1 & 3 & 5 & 3 & 2 & 1 & 5 & 3 & 6 & 2 & 5 & 3 & 7 & 9 & 5 & 3 & 8 & 7 & 5 & 3 & 10 & 3 & 5 & 4 & 4 & 2 \\
5 & 4 & 5 & 9 & 5 & 4 & 6 & 7 & 5 & 4 & 8 & 3 & 5 & 4 & 9 & 1 & 5 & 5 & 0 & 2 & 5 & 5 & 2 & 2 & 5 & 5 & 3 & 9 & 5 & 5 & 4 & 7 & 5 & 5 & 7 & 1 & 5 & 6 & 0 & 7 & 5 & 6 & 1 & 9 \\
5 & 6 & 2 & 7 & 5 & 6 & 5 & 1 & 5 & 6 & 9 & 2 & 5 & 6 & 10 & 9 & 5 & 7 & 0 & 3 & 5 & 7 & 3 & 1 & 5 & 7 & 7 & 2 & 5 & 7 & 8 & 9 & 5 & 7 & 9 & 7 & 5 & 8 & 1 & 1 & 5 & 8 & 5 & 2 \\
5 & 8 & 6 & 9 & 5 & 8 & 7 & 7 & 5 & 8 & 9 & 3 & 5 & 8 & 10 & 1 & 5 & 9 & 3 & 2 & 5 & 9 & 4 & 9 & 5 & 9 & 5 & 7 & 5 & 9 & 7 & 3 & 5 & 9 & 8 & 1 & 5 & 10 & 1 & 2 & 5 & 10 & 2 & 9 \\
5 & 10 & 3 & 7 & 5 & 10 & 5 & 3 & 5 & 10 & 6 & 1 & 5 & 10 & 10 & 2 & 6 & 0 & 0 & 0 & 6 & 0 & 0 & 4 & 6 & 0 & 2 & 5 & 6 & 0 & 5 & 7 & 6 & 0 & 6 & 5 & 6 & 0 & 7 & 3 & 6 & 0 & 9 & 8 \\
6 & 0 & 9 & 10 & 6 & 1 & 2 & 0 & 6 & 1 & 4 & 5 & 6 & 1 & 5 & 3 & 6 & 1 & 7 & 8 & 6 & 1 & 7 & 10 & 6 & 1 & 8 & 6 & 6 & 1 & 9 & 4 & 6 & 2 & 0 & 5 & 6 & 2 & 2 & 5 & 6 & 2 & 5 & 8 \\
6 & 2 & 5 & 10 & 6 & 2 & 6 & 6 & 6 & 2 & 7 & 4 & 6 & 2 & 9 & 0 & 6 & 3 & 3 & 8 & 6 & 3 & 3 & 10 & 6 & 3 & 4 & 6 & 6 & 3 & 5 & 4 & 6 & 3 & 7 & 0 & 6 & 3 & 9 & 5 & 6 & 4 & 0 & 6 \\
6 & 4 & 1 & 8 & 6 & 4 & 1 & 10 & 6 & 4 & 2 & 6 & 6 & 4 & 3 & 4 & 6 & 4 & 5 & 0 & 6 & 4 & 7 & 5 & 6 & 4 & 10 & 8 & 6 & 4 & 10 & 10 & 6 & 5 & 1 & 4 & 6 & 5 & 3 & 0 & 6 & 5 & 5 & 5 \\
6 & 5 & 6 & 3 & 6 & 5 & 8 & 8 & 6 & 5 & 8 & 10 & 6 & 5 & 9 & 6 & 6 & 5 & 10 & 4 & 6 & 6 & 1 & 0 & 6 & 6 & 3 & 5 & 6 & 6 & 4 & 3 & 6 & 6 & 6 & 8 & 6 & 6 & 6 & 10 & 6 & 6 & 7 & 6 \\
6 & 6 & 8 & 4 & 6 & 6 & 10 & 0 & 6 & 7 & 1 & 5 & 6 & 7 & 2 & 3 & 6 & 7 & 4 & 8 & 6 & 7 & 4 & 10 & 6 & 7 & 5 & 6 & 6 & 7 & 6 & 4 & 6 & 7 & 8 & 0 & 6 & 7 & 10 & 5 & 6 & 8 & 2 & 8 \\
6 & 8 & 2 & 10 & 6 & 8 & 3 & 6 & 6 & 8 & 4 & 4 & 6 & 8 & 6 & 0 & 6 & 8 & 8 & 5 & 6 & 9 & 0 & 4 & 6 & 9 & 0 & 8 & 6 & 9 & 0 & 10 & 6 & 9 & 1 & 6 & 6 & 9 & 2 & 4 & 6 & 9 & 4 & 0 \\
6 & 9 & 6 & 5 & 6 & 9 & 9 & 8 & 6 & 9 & 9 & 10 & 6 & 9 & 10 & 6 & 6 & 10 & 2 & 0 & 6 & 10 & 4 & 5 & 6 & 10 & 7 & 8 & 6 & 10 & 7 & 10 & 6 & 10 & 8 & 6 & 6 & 10 & 9 & 4 & 7 & 0 & 1 & 7 \\
7 & 0 & 1 & 9 & 7 & 0 & 3 & 2 & 7 & 0 & 4 & 0 & 7 & 0 & 5 & 9 & 7 & 0 & 8 & 1 & 7 & 0 & 10 & 6 & 7 & 1 & 1 & 2 & 7 & 1 & 3 & 7 & 7 & 1 & 3 & 9 & 7 & 1 & 6 & 1 & 7 & 1 & 10 & 2 \\
7 & 2 & 0 & 0 & 7 & 2 & 1 & 7 & 7 & 2 & 1 & 9 & 7 & 2 & 3 & 3 & 7 & 2 & 4 & 1 & 7 & 2 & 8 & 2 & 7 & 2 & 10 & 7 & 7 & 2 & 10 & 9 & 7 & 3 & 1 & 3 & 7 & 3 & 2 & 1 & 7 & 3 & 6 & 2 \\
7 & 3 & 8 & 7 & 7 & 3 & 8 & 9 & 7 & 3 & 10 & 3 & 7 & 4 & 0 & 1 & 7 & 4 & 4 & 2 & 7 & 4 & 6 & 7 & 7 & 4 & 6 & 9 & 7 & 4 & 8 & 3 & 7 & 4 & 9 & 1 & 7 & 5 & 2 & 2 & 7 & 5 & 4 & 7 \\
7 & 5 & 4 & 9 & 7 & 5 & 7 & 1 & 7 & 6 & 0 & 2 & 7 & 6 & 2 & 7 & 7 & 6 & 2 & 9 & 7 & 6 & 5 & 1 & 7 & 6 & 9 & 2 & 7 & 7 & 0 & 7 & 7 & 7 & 0 & 9 & 7 & 7 & 3 & 1 & 7 & 7 & 7 & 2 \\
7 & 7 & 9 & 7 & 7 & 7 & 9 & 9 & 7 & 8 & 1 & 1 & 7 & 8 & 5 & 2 & 7 & 8 & 7 & 7 & 7 & 8 & 7 & 9 & 7 & 8 & 9 & 3 & 7 & 8 & 10 & 1 & 7 & 9 & 3 & 2 & 7 & 9 & 5 & 7 & 7 & 9 & 5 & 9 \\
7 & 9 & 7 & 3 & 7 & 9 & 8 & 1 & 7 & 10 & 1 & 2 & 7 & 10 & 3 & 7 & 7 & 10 & 3 & 9 & 7 & 10 & 5 & 3 & 7 & 10 & 6 & 1 & 7 & 10 & 10 & 2 & 8 & 0 & 2 & 0 & 8 & 0 & 2 & 4 & 8 & 0 & 6 & 5 \\
8 & 0 & 6 & 7 & 8 & 0 & 7 & 3 & 8 & 0 & 10 & 8 & 8 & 0 & 10 & 10 & 8 & 1 & 0 & 4 & 8 & 1 & 4 & 5 & 8 & 1 & 5 & 3 & 8 & 1 & 8 & 6 & 8 & 1 & 8 & 8 & 8 & 1 & 8 & 10 & 8 & 1 & 9 & 4 \\
8 & 2 & 2 & 5 & 8 & 2 & 6 & 6 & 8 & 2 & 6 & 8 & 8 & 2 & 6 & 10 & 8 & 2 & 7 & 4 & 8 & 3 & 0 & 5 & 8 & 3 & 4 & 6 & 8 & 3 & 4 & 8 & 8 & 3 & 4 & 10 & 8 & 3 & 5 & 4 & 8 & 3 & 9 & 5 \\
8 & 4 & 2 & 6 & 8 & 4 & 2 & 8 & 8 & 4 & 2 & 10 & 8 & 4 & 3 & 4 & 8 & 4 & 7 & 5 & 8 & 5 & 0 & 6 & 8 & 5 & 0 & 8 & 8 & 5 & 0 & 10 & 8 & 5 & 1 & 4 & 8 & 5 & 5 & 5 & 8 & 5 & 6 & 3 \\
8 & 5 & 9 & 6 & 8 & 5 & 9 & 8 & 8 & 5 & 9 & 10 & 8 & 5 & 10 & 4 & 8 & 6 & 3 & 5 & 8 & 6 & 4 & 3 & 8 & 6 & 7 & 6 & 8 & 6 & 7 & 8 & 8 & 6 & 7 & 10 & 8 & 6 & 8 & 4 & 8 & 7 & 1 & 5 \\
8 & 7 & 2 & 3 & 8 & 7 & 5 & 6 & 8 & 7 & 5 & 8 & 8 & 7 & 5 & 10 & 8 & 7 & 6 & 4 & 8 & 7 & 10 & 5 & 8 & 8 & 0 & 3 & 8 & 8 & 3 & 6 & 8 & 8 & 3 & 8 & 8 & 8 & 3 & 10 & 8 & 8 & 4 & 4 \\
8 & 8 & 8 & 5 & 8 & 9 & 1 & 6 & 8 & 9 & 1 & 8 & 8 & 9 & 1 & 10 & 8 & 9 & 2 & 4 & 8 & 9 & 6 & 5 & 8 & 9 & 10 & 6 & 8 & 9 & 10 & 8 & 8 & 9 & 10 & 10 & 8 & 10 & 0 & 4 & 8 & 10 & 4 & 5 \\
8 & 10 & 8 & 6 & 8 & 10 & 8 & 8 & 8 & 10 & 8 & 10 & 8 & 10 & 9 & 4 & 9 & 0 & 0 & 1 & 9 & 0 & 2 & 2 & 9 & 0 & 2 & 6 & 9 & 0 & 2 & 8 & 9 & 0 & 6 & 0 & 9 & 0 & 6 & 9 & 9 & 1 & 0 & 6 \\
9 & 1 & 4 & 0 & 9 & 1 & 4 & 7 & 9 & 1 & 4 & 9 & 9 & 1 & 9 & 1 & 9 & 2 & 2 & 0 & 9 & 2 & 2 & 7 & 9 & 2 & 2 & 9 & 9 & 2 & 3 & 2 & 9 & 2 & 7 & 1 & 9 & 3 & 0 & 0 & 9 & 3 & 0 & 7 \\
9 & 3 & 0 & 9 & 9 & 3 & 1 & 2 & 9 & 3 & 5 & 1 & 9 & 3 & 9 & 0 & 9 & 3 & 9 & 7 & 9 & 3 & 9 & 9 & 9 & 3 & 10 & 2 & 9 & 4 & 3 & 1 & 9 & 4 & 7 & 0 & 9 & 4 & 7 & 7 & 9 & 4 & 7 & 9 \\
9 & 4 & 8 & 2 & 9 & 5 & 1 & 1 & 9 & 5 & 5 & 1 & 9 & 5 & 5 & 8 & 9 & 5 & 5 & 10 & 9 & 5 & 10 & 1 & 9 & 6 & 3 & 1 & 9 & 6 & 3 & 8 & 9 & 6 & 3 & 10 & 9 & 6 & 8 & 1 & 9 & 7 & 1 & 0 \\
9 & 7 & 1 & 7 & 9 & 7 & 1 & 9 & 9 & 7 & 6 & 1 & 9 & 7 & 10 & 0 & 9 & 7 & 10 & 7 & 9 & 7 & 10 & 9 & 9 & 8 & 4 & 1 & 9 & 8 & 8 & 0 & 9 & 8 & 8 & 2 & 9 & 8 & 8 & 7 & 9 & 8 & 8 & 9 \\
9 & 9 & 0 & 1 & 9 & 9 & 2 & 1 & 9 & 9 & 6 & 0 & 9 & 9 & 6 & 2 & 9 & 9 & 6 & 7 & 9 & 9 & 6 & 9 & 9 & 10 & 4 & 0 & 9 & 10 & 4 & 2 & 9 & 10 & 4 & 7 & 9 & 10 & 4 & 9 & 9 & 10 & 8 & 1 \\
10 & 0 & 2 & 4 & 10 & 0 & 2 & 10 & 10 & 0 & 7 & 3 & 10 & 0 & 7 & 5 & 10 & 0 & 7 & 7 & 10 & 1 & 0 & 3 & 10 & 1 & 0 & 8 & 10 & 1 & 0 & 10 & 10 & 1 & 5 & 2 & 10 & 1 & 5 & 4 & 10 & 1 & 9 & 3 \\
10 & 1 & 9 & 5 & 10 & 1 & 9 & 7 & 10 & 1 & 9 & 9 & 10 & 2 & 3 & 4 & 10 & 2 & 7 & 3 & 10 & 2 & 7 & 5 & 10 & 2 & 7 & 7 & 10 & 2 & 7 & 9 & 10 & 3 & 1 & 4 & 10 & 3 & 5 & 3 & 10 & 3 & 5 & 5 \\
10 & 3 & 5 & 7 & 10 & 3 & 5 & 9 & 10 & 3 & 10 & 4 & 10 & 4 & 3 & 3 & 10 & 4 & 3 & 5 & 10 & 4 & 3 & 7 & 10 & 4 & 3 & 9 & 10 & 4 & 8 & 4 & 10 & 5 & 1 & 3 & 10 & 5 & 1 & 5 & 10 & 5 & 1 & 7 \\
10 & 5 & 1 & 9 & 10 & 5 & 6 & 3 & 10 & 5 & 6 & 5 & 10 & 5 & 10 & 3 & 10 & 5 & 10 & 7 & 10 & 5 & 10 & 9 & 10 & 6 & 4 & 3 & 10 & 6 & 4 & 5 & 10 & 6 & 8 & 3 & 10 & 6 & 8 & 5 & 10 & 6 & 8 & 7 \\
10 & 6 & 8 & 9 & 10 & 6 & 10 & 5 & 10 & 7 & 1 & 2 & 10 & 7 & 2 & 4 & 10 & 7 & 6 & 3 & 10 & 7 & 6 & 5 & 10 & 7 & 6 & 7 & 10 & 7 & 6 & 9 & 10 & 7 & 10 & 2 & 10 & 8 & 0 & 4 & 10 & 8 & 4 & 3 \\
10 & 8 & 4 & 5 & 10 & 8 & 4 & 7 & 10 & 8 & 4 & 9 & 10 & 8 & 9 & 4 & 10 & 9 & 2 & 3 & 10 & 9 & 2 & 5 & 10 & 9 & 2 & 7 & 10 & 9 & 2 & 9 & 10 & 9 & 6 & 4 & 10 & 10 & 0 & 3 & 10 & 10 & 0 & 5 \\
10 & 10 & 0 & 7 & 10 & 10 & 0 & 9 & 10 & 10 & 4 & 4 & 10 & 10 & 9 & 3 & 10 & 10 & 9 & 5 & 10 & 10 & 9 & 7 & 10 & 10 & 9 & 9 &  &  &  &  &  &  &  &  &  &  &  &  &  &  &  &  \\
\hline
\end{tabular}%
}
\caption{Independent set of size $766$ in $C_{11}^4$.}
\label{tab:c11-766}
\end{table}

\clearpage
\subsection{\texorpdfstring{$C_{15}^3$}{C15, power 3}}
\label{appendix:c15-383}

\begin{table}[H]
\centering
\footnotesize
\setlength{\tabcolsep}{2pt}
\renewcommand{\arraystretch}{0.68}
\resizebox{0.88\textwidth}{!}{%
\begin{tabular}{|*{7}{rrr|}}
\hline
8 & 0 & 0 & 3 & 1 & 0 & 14 & 1 & 0 & 7 & 2 & 0 & 2 & 3 & 0 & 13 & 3 & 0 & 4 & 4 & 0 \\
10 & 5 & 0 & 12 & 5 & 0 & 3 & 6 & 0 & 7 & 6 & 0 & 13 & 7 & 0 & 2 & 8 & 0 & 6 & 8 & 0 \\
12 & 9 & 0 & 1 & 10 & 0 & 5 & 10 & 0 & 9 & 11 & 0 & 11 & 11 & 0 & 0 & 12 & 0 & 5 & 12 & 0 \\
9 & 13 & 0 & 11 & 13 & 0 & 0 & 14 & 0 & 4 & 14 & 0 & 6 & 0 & 1 & 1 & 1 & 1 & 10 & 1 & 1 \\
12 & 1 & 1 & 5 & 2 & 1 & 0 & 3 & 1 & 9 & 3 & 1 & 11 & 3 & 1 & 6 & 4 & 1 & 14 & 5 & 1 \\
1 & 6 & 1 & 5 & 6 & 1 & 9 & 7 & 1 & 11 & 7 & 1 & 0 & 8 & 1 & 4 & 8 & 1 & 8 & 9 & 1 \\
10 & 9 & 1 & 3 & 10 & 1 & 14 & 10 & 1 & 7 & 11 & 1 & 3 & 12 & 1 & 13 & 12 & 1 & 7 & 13 & 1 \\
2 & 14 & 1 & 13 & 14 & 1 & 4 & 0 & 2 & 8 & 1 & 2 & 14 & 1 & 2 & 3 & 2 & 2 & 13 & 3 & 2 \\
2 & 4 & 2 & 4 & 4 & 2 & 8 & 5 & 2 & 10 & 5 & 2 & 12 & 5 & 2 & 3 & 6 & 2 & 7 & 7 & 2 \\
2 & 8 & 2 & 13 & 8 & 2 & 6 & 9 & 2 & 1 & 10 & 2 & 12 & 10 & 2 & 5 & 11 & 2 & 1 & 12 & 2 \\
9 & 12 & 2 & 11 & 12 & 2 & 5 & 13 & 2 & 0 & 14 & 2 & 9 & 14 & 2 & 11 & 14 & 2 & 2 & 0 & 3 \\
6 & 1 & 3 & 10 & 1 & 3 & 12 & 1 & 3 & 1 & 2 & 3 & 7 & 3 & 3 & 9 & 3 & 3 & 11 & 3 & 3 \\
0 & 4 & 3 & 6 & 5 & 3 & 1 & 6 & 3 & 14 & 6 & 3 & 5 & 7 & 3 & 0 & 8 & 3 & 9 & 8 & 3 \\
11 & 8 & 3 & 4 & 9 & 3 & 8 & 10 & 3 & 10 & 10 & 3 & 14 & 10 & 3 & 3 & 11 & 3 & 7 & 12 & 3 \\
13 & 12 & 3 & 3 & 13 & 3 & 7 & 14 & 3 & 13 & 14 & 3 & 0 & 0 & 4 & 4 & 1 & 4 & 8 & 1 & 4 \\
14 & 2 & 4 & 3 & 3 & 4 & 5 & 3 & 4 & 13 & 4 & 4 & 4 & 5 & 4 & 8 & 6 & 4 & 10 & 6 & 4 \\
12 & 6 & 4 & 3 & 7 & 4 & 7 & 8 & 4 & 13 & 8 & 4 & 2 & 9 & 4 & 6 & 10 & 4 & 12 & 10 & 4 \\
1 & 11 & 4 & 5 & 12 & 4 & 9 & 12 & 4 & 11 & 12 & 4 & 1 & 13 & 4 & 5 & 14 & 4 & 9 & 14 & 4 \\
11 & 14 & 4 & 13 & 0 & 5 & 2 & 1 & 5 & 6 & 1 & 5 & 10 & 2 & 5 & 12 & 2 & 5 & 1 & 3 & 5 \\
7 & 4 & 5 & 9 & 4 & 5 & 11 & 4 & 5 & 0 & 5 & 5 & 2 & 5 & 5 & 6 & 6 & 5 & 1 & 7 & 5 \\
5 & 8 & 5 & 9 & 8 & 5 & 11 & 8 & 5 & 0 & 9 & 5 & 4 & 10 & 5 & 8 & 10 & 5 & 10 & 10 & 5 \\
14 & 11 & 5 & 3 & 12 & 5 & 7 & 12 & 5 & 14 & 13 & 5 & 3 & 14 & 5 & 7 & 14 & 5 & 0 & 0 & 6 \\
9 & 0 & 6 & 11 & 0 & 6 & 4 & 1 & 6 & 8 & 2 & 6 & 14 & 2 & 6 & 3 & 3 & 6 & 5 & 3 & 6 \\
13 & 4 & 6 & 4 & 5 & 6 & 8 & 6 & 6 & 10 & 6 & 6 & 12 & 6 & 6 & 3 & 7 & 6 & 14 & 7 & 6 \\
7 & 8 & 6 & 2 & 9 & 6 & 13 & 9 & 6 & 6 & 10 & 6 & 1 & 11 & 6 & 12 & 11 & 6 & 5 & 12 & 6 \\
1 & 13 & 6 & 10 & 13 & 6 & 12 & 13 & 6 & 5 & 14 & 6 & 7 & 0 & 7 & 13 & 0 & 7 & 2 & 1 & 7 \\
10 & 2 & 7 & 12 & 2 & 7 & 1 & 3 & 7 & 7 & 4 & 7 & 9 & 4 & 7 & 11 & 4 & 7 & 0 & 5 & 7 \\
2 & 5 & 7 & 6 & 6 & 7 & 1 & 7 & 7 & 5 & 8 & 7 & 0 & 9 & 7 & 9 & 9 & 7 & 11 & 9 & 7 \\
4 & 10 & 7 & 8 & 11 & 7 & 10 & 11 & 7 & 14 & 11 & 7 & 3 & 12 & 7 & 8 & 13 & 7 & 14 & 13 & 7 \\
3 & 14 & 7 & 5 & 0 & 8 & 9 & 0 & 8 & 11 & 0 & 8 & 0 & 1 & 8 & 4 & 2 & 8 & 6 & 2 & 8 \\
8 & 2 & 8 & 14 & 3 & 8 & 5 & 4 & 8 & 13 & 5 & 8 & 4 & 6 & 8 & 8 & 7 & 8 & 10 & 7 & 8 \\
12 & 7 & 8 & 14 & 7 & 8 & 3 & 8 & 8 & 7 & 9 & 8 & 13 & 9 & 8 & 2 & 10 & 8 & 6 & 11 & 8 \\
12 & 11 & 8 & 1 & 12 & 8 & 6 & 13 & 8 & 10 & 13 & 8 & 12 & 13 & 8 & 1 & 14 & 8 & 3 & 0 & 9 \\
7 & 0 & 9 & 13 & 1 & 9 & 2 & 2 & 9 & 10 & 3 & 9 & 12 & 3 & 9 & 1 & 4 & 9 & 3 & 4 & 9 \\
7 & 5 & 9 & 9 & 5 & 9 & 11 & 5 & 9 & 2 & 6 & 9 & 6 & 7 & 9 & 1 & 8 & 9 & 5 & 9 & 9 \\
9 & 9 & 9 & 11 & 9 & 9 & 0 & 10 & 9 & 4 & 11 & 9 & 8 & 11 & 9 & 10 & 11 & 9 & 14 & 12 & 9 \\
4 & 13 & 9 & 8 & 13 & 9 & 14 & 14 & 9 & 1 & 0 & 10 & 5 & 0 & 10 & 9 & 1 & 10 & 11 & 1 & 10 \\
0 & 2 & 10 & 4 & 2 & 10 & 6 & 3 & 10 & 8 & 3 & 10 & 14 & 4 & 10 & 5 & 5 & 10 & 0 & 6 & 10 \\
13 & 6 & 10 & 4 & 7 & 10 & 8 & 7 & 10 & 10 & 7 & 10 & 14 & 8 & 10 & 3 & 9 & 10 & 7 & 9 & 10 \\
13 & 10 & 10 & 2 & 11 & 10 & 6 & 11 & 10 & 12 & 12 & 10 & 2 & 13 & 10 & 6 & 13 & 10 & 10 & 14 & 10 \\
12 & 14 & 10 & 3 & 0 & 11 & 7 & 1 & 11 & 13 & 1 & 11 & 2 & 2 & 11 & 10 & 3 & 11 & 12 & 3 & 11 \\
1 & 4 & 11 & 3 & 4 & 11 & 7 & 5 & 11 & 9 & 5 & 11 & 11 & 5 & 11 & 2 & 6 & 11 & 6 & 7 & 11 \\
1 & 8 & 11 & 12 & 8 & 11 & 5 & 9 & 11 & 0 & 10 & 11 & 9 & 10 & 11 & 11 & 10 & 11 & 4 & 11 & 11 \\
0 & 12 & 11 & 8 & 12 & 11 & 10 & 12 & 11 & 4 & 13 & 11 & 8 & 14 & 11 & 14 & 14 & 11 & 1 & 0 & 12 \\
5 & 1 & 12 & 9 & 1 & 12 & 11 & 1 & 12 & 0 & 2 & 12 & 6 & 3 & 12 & 8 & 3 & 12 & 14 & 4 & 12 \\
5 & 5 & 12 & 0 & 6 & 12 & 13 & 6 & 12 & 4 & 7 & 12 & 8 & 8 & 12 & 10 & 8 & 12 & 14 & 8 & 12 \\
3 & 9 & 12 & 7 & 10 & 12 & 13 & 10 & 12 & 2 & 11 & 12 & 6 & 12 & 12 & 13 & 12 & 12 & 2 & 13 & 12 \\
6 & 14 & 12 & 10 & 14 & 12 & 12 & 14 & 12 & 14 & 0 & 13 & 3 & 1 & 13 & 7 & 1 & 13 & 13 & 2 & 13 \\
2 & 3 & 13 & 4 & 3 & 13 & 10 & 4 & 13 & 12 & 4 & 13 & 3 & 5 & 13 & 7 & 6 & 13 & 9 & 6 & 13 \\
11 & 6 & 13 & 2 & 7 & 13 & 6 & 8 & 13 & 12 & 8 & 13 & 1 & 9 & 13 & 5 & 10 & 13 & 9 & 10 & 13 \\
11 & 10 & 13 & 0 & 11 & 13 & 4 & 12 & 13 & 9 & 12 & 13 & 11 & 12 & 13 & 0 & 13 & 13 & 4 & 14 & 13 \\
8 & 14 & 13 & 10 & 0 & 14 & 12 & 0 & 14 & 1 & 1 & 14 & 5 & 1 & 14 & 9 & 2 & 14 & 11 & 2 & 14 \\
0 & 3 & 14 & 6 & 4 & 14 & 8 & 4 & 14 & 1 & 5 & 14 & 14 & 5 & 14 & 5 & 6 & 14 & 0 & 7 & 14 \\
4 & 8 & 14 & 8 & 8 & 14 & 10 & 8 & 14 & 14 & 9 & 14 & 3 & 10 & 14 & 7 & 10 & 14 & 13 & 11 & 14 \\
2 & 12 & 14 & 7 & 12 & 14 & 13 & 13 & 14 & 2 & 14 & 14 & 6 & 14 & 14 &  &  &  &  &  &  \\
\hline
\end{tabular}%
}
\caption{Independent set of size $383$ in $C_{15}^{3}$.}
\label{tab:c15-383}
\end{table}
\end{document}